\documentclass{llncs}

\usepackage{amssymb}
\usepackage{amsmath}
\usepackage{textcomp}
\usepackage{array}
\usepackage{lineno}
\usepackage{color}
\usepackage{comment}

\usepackage{graphicx}

\usepackage{listings}
\usepackage[colorinlistoftodos]{todonotes}
\usepackage[normalem]{ulem} 

\newcommand{\cent}[0]{\mbox{\textcent}}
\newcommand{\dollar}[0]{\$}

\newtheorem{fact}{Fact}

\definecolor{CLNOTE}{rgb}{0,0,.6}
\definecolor{CLNOTERED}{rgb}{.8,0,0}
\definecolor{CLTYPO}{rgb}{.8,0,.8}

\title{Unary languages recognized by two-way one-counter automata}

\author{Marzio De Biasi\inst{1} \and  Abuzer Yakary{\i}lmaz\inst{2,3}$^,$\thanks{Yakary{\i}lmaz was partially supported by CAPES, ERC Advanced Grant MQC, and FP7 FET project QALGO.}}

\institute{\email{marziodebiasi@gmail.com}
\and
University of Latvia, Faculty of Computing, Raina bulv. 19, R\={\i}ga, LV-1586, Latvia 
\and
National Laboratory for Scientific Computing, Petr\'{o}polis, RJ, 25651-075, Brazil \\
\email{abuzer@lncc.br}
}

\authorrunning{De Biasi and Yakary{\i}lmaz} 

\definecolor{gray}{RGB}{128,128,128}
\definecolor{reg}{RGB}{255,0,0}

\begin{document}

\maketitle

\begin{abstract}
A two-way deterministic finite state automaton with  one counter (2D1CA) is a fundamental computational model that has been examined in many different aspects since sixties,
but we know little about its power in the case of unary languages.
Up to our knowledge, the only known unary nonregular languages recognized by 2D1CAs are those formed by strings having exponential length, where the exponents form some trivial unary regular language. In this paper, we present some non-trivial subsets of these languages. By using the input head as a second counter, we present  simulations of two-way deterministic finite automata with linearly bounded counters and linear--space Turing machines. We also show how a fixed-size quantum register can help to simplify some of these languages.  Finally, we compare unary 2D1CAs with two--counter machines and provide some insights about the limits of their computational power. 

\end{abstract}

\section{Introduction}

A finite automaton with one-counter is a fundamental model in automata theory. It has been examined in many different aspects since sixties
\cite{FMR68}. One recent significant result, for example, is that the equivalence problem of deterministic one-way counter automata is NL-complete \cite{BGJ13}. After introducing quantum automata \cite{MC00,KW97} at the end of the nineties, quantum counter automata have also been examined (see a very recent research work in \cite{Yak13B}). 

 A \emph{counter} is a very simple working memory  which can store an arbitrary long integer that can be incremented or decremented; but only a single bit of information can be retrieved from it: whether its value is zero or not. It is a well-known fact that a two-way deterministic finite automaton with two counters is universal \cite{Min61,Min67,Mor96}. Any language recognized by a two-way deterministic finite automaton with one counter (2D1CA), on the other hand, is in deterministic logarithmic space ($ \mathsf{L} $) \cite{Pet12A}.\footnote{\label{foot:2D1CA}Since a 2D1CA using super-linear space on its counter should finally enter an infinite loop, any useful algorithm can use at most linear space on a counter, which can be simulated by a logarithmic binary working tape.} Replacing the counter of a 2D1CA with a stack, we get a  two-way deterministic pushdown automaton (2DPDA), that can recognize more languages \cite{DG82A}. Similarly, nondeterminism also increases the class of the languages recognized by 2D1CAs \cite{Chr84}. 

Unary or tally languages, defined over a single letter alphabet, have deserved special attention.  When the input head is not allowed to move to the left (\emph{one-way head}), it is a well-known fact that unary nondeterministic pushdown automata can recognize only  regular languages \cite{GR62}. The same result was shown for bounded-error probabilistic pushdown automata, too \cite{KGF97}. Currently, we do not know whether ``quantumness'' can add any power. Their alternating versions were shown to be quite powerful: they can recognize any unary language in deterministic exponential time with linear exponent \cite{CKS81}. But, if we replace the stack with a counter, only a single family of unary nonregular languages \cite{Dur13A} is known:
$
	\mathtt{UPOWER(k)} = \left\lbrace a^{k^n} \mid n \geq 1 \right\rbrace ,
$ for a given integer $ k \geq 2 $.
In the case of two-way head, we know that the unary encoding of every language in deterministic polynomial time ($ \mathsf{P} $) can be accepted by 2DPDAs \cite{Mon84}; however we do not know whether 2DPDAs are more powerful than 2D1CAs (see also \cite{IJTW93}) on unary languages like in the case of binary languages. Any separation between $ \mathsf{L} $ and $ \mathsf{P} $ can of course answer this question positively, but, it is still one of the big open problems in complexity theory. On the other hand, researchers also proposed some simple candidate languages not seemingly accepted by any 2D1CA \cite{IJTW93,Pet94}, e.g.
$
	\mathtt{USQUARE} = \left\lbrace a^{n^2} \mid n \geq 1 \right\rbrace.
$
Although it was shown that two--counter machines (2CAs) cannot recognize $\tt USQUARE $ if the input counter is initialized with $n^2$ (i.e. no G\"odelization
is allowed) \cite{IT93,Sch72}, up to our knowledge, there is not any known nondeterministic, alternating, or probabilistic one-counter automaton for it.
We only know that $ \mathtt{USQUARE} $ can be recognized by exponential expected time 2D1CAs augmented with a fixed-size quantum register \cite{Yak13B} or realtime private alternating one-counter automata \cite{DHRSY14}. Apart from this open problem, we do not know much about which nonregular unary languages can be recognized by 2D1CAs. In this paper, we provide some answers to this question. 
In his seminal paper \cite{Min61}, Minsky showed that the emptiness problem for 2D1CAs on unary languages is undecidable. In his proof, he presented a simulation of two-way deterministic finite automaton with two counters on the empty string by a 2D1CA using its input head as a second counter. We use a similar idea but as a new programming technique for 2D1CAs on unary languages that allows to simulate multi-counter automata and space bounded Turing machines operating on unary or general alphabets. A 2D1CA can take the input and the working memory of the simulated machine as the exponent of some integers encoded on unary inputs. Thus, once the automaton becomes sure about the correctness of the encoding, it can start a two-counter simulation of the given machine, in which the second counter is implemented by the  head on the unary input. Based on this idea, we will present several new nonregular unary languages recognized by 2D1CAs. Our technique can be applicable to nondeterministic, alternating, and probabilistic cases in a straightforward way. We also show that using a constant-size quantum memory can help to replace the encoding on binary alphabets with unary alphabets. Finally we compare unary 2D1CAs with 2CAs and provide some insights about the limits of their computational power.

\section{Background}

Throughout the paper, $ \Sigma $  denotes the input alphabet and  the extra symbols  $ \cent $ and $ \dollar $ are the \textit{end-markers} (the tape alphabet is $ \tilde{\Sigma} = \Sigma \cup \{ \cent,\dollar \} $). For a given string $ w $, $ w^r $ is the reverse of $ w $, $ |w| $ is the length of $ w $, and $ w_{i} $ is the $ i^{th} $ symbol of $ w $, where $ 1 \leq i \leq |w| $. The string $ \cent w \dollar $ is represented by $ \tilde{w} $. 
Each counter model defined in the paper has a two-way finite read-only input tape whose squares are indexed by integers. Any given input string, say $ w \in \Sigma^{*} $, is placed on the tape as $ \tilde{w} $ between the squares indexed by 0 and $ | w | +1 $. The tape has a single head, and it can stay in the same position ($ \downarrow $) or move to one square to the left ($ \leftarrow $) or to the right ($ \rightarrow $) in one step. It must always be guaranteed that the input head never leaves $ \tilde{w} $. A counter can store an integer and has two observable states: \textit{zero} ($ 0 $) or \textit{nonzero} ($ \pm $), and can be updated by a value from $ \{-1,0,+1\} $ in one step. Let $\Theta = \{ 0, \pm \}$.

A \textit{two-way deterministic one-counter automaton} (2D1CA) is a two-way deterministic finite automaton with a counter. Formally, a 2D1CA $ \mathcal{D} $ is a 6-tuple
\[
	\mathcal{D} = (S,\Sigma,\delta,s_{1},s_{a},s_{r}) ,
\] 
where $ S $ is the set of states, $ s_1 \in S $ is the initial state, $ s_a,s_r \in S $ ($ s_a \neq s_r $) are the accepting and rejecting states, respectively, and $ \delta $  is the transition function governing the behaviour of $ \mathcal{D} $ in each step, i.e. 
\[
	\delta: S \setminus \{s_a,s_r\} \times \tilde{\Sigma} \times \Theta \rightarrow S \times \mspace{-5mu} \{\leftarrow,\downarrow,\rightarrow\} \mspace{-5mu} \times \{-1,0,+1\}. 
\]
Specifically, $ \delta(s,\sigma,\theta) \rightarrow (s',d_i,c) $ means that when $ \mathcal{D} $ is in state $ s \in S \setminus \{s_a,s_r\} $, reads symbol $ \sigma \in \tilde{\Sigma} $, and the state of its counter is $ \theta \in \Theta $, then it updates its state to $ s' \in S $ and the position of the input head with respect to $ d_i \in \{\leftarrow,\downarrow,\rightarrow\} $, and adds $ c \in \{-1,0,+1\} $ to the value of the counter. In order to stay on the boundaries of $\tilde w$, if $ \sigma= \cent$ then $d_i \in \{\downarrow, \rightarrow \}$ and if $\sigma= \dollar$ then $d_i \in \{\downarrow, \leftarrow \}$.

At the beginning of the computation, $ \mathcal{D} $ is in state $ s_1 $, the input head is placed on symbol $ \cent $, and the value of the counter is set to zero. A configuration of $ \mathcal{D} $ on a given input string is represented by a triple $ (s,i,v) $, where $ s $ is the state, $ i $ is the position of the input head, and $ v $ is the value of the counter. The computation ends and the input is accepted (resp. rejected) by $ \mathcal{D} $ when it enters $ s_a $ (resp. $ s_r $). 

For any $k>1$, a \textit{two-way deterministic $k$-counter automaton} (2D$k$CA) is a generalization of a 2D1CA and is equipped with $k$ counters; in each transition, it checks the  states of all counters and then updates their values.
Moreover, we call a counter \textit{linearly bounded} if its value never exceeds $O(|w|)$, where $w$ is the given input. But restricting this bound to $|w|$ does not change the computational power of any kind of automaton having linearly bounded counters, i.e. the value of any counter can be compressed by any rational number by using extra control states. A \emph{two-counter automaton} (2CA) is a 2D2CA  over a unary alphabet and without the input tape: the length of the unary input is placed  in one of the counters at the beginning of the computation. We underline that
the length of the unary input $a^n$ is placed in the counter as it is: indeed if we allow a suitable
encoding of the input (by G\"{o}delization, e.g. setting its initial value to $2^n$) a 2CA can simulate any Turing machine \cite{Min61,Sch72}.

We replace ``D'' that stands for \textit{deterministic} in the abbreviations of deterministic machines with ``N'', ``A'', and ``P'' for representing the abbreviations of their \textit{nondeterministic}, \textit{alternating}, and \textit{probabilistic} counterparts. 

We finish the section with some useful technical lemmas.


\begin{lemma}
	\label{lemma:UPOWER}
	2D1CAs can check whether a given string is a member of language $ \mathtt{UPOWER(k)}=\{a^{k^n} \mid n \geq 1\} $, with $k\geq 2$.
\end{lemma}

\begin{lemma}
	\label{lemma:divide-by-p}
	For any given $p \in \mathbb{Z}^+ $, there exists a 2D1CA $ \mathcal{D} $ that can set the value of its counter to $ M $ if its initial value is $ M \cdot p^n $  provided that the length of the input is at least $ M \cdot p^{n-1} $, where $ M \in \mathbb{Z}^+ $, $p \nmid M$, and $n>0$.
\end{lemma}


\begin{lemma}
\label{lemma:coprimespow}
	The language $L = \{ a^{2^j 3^k} \mid j,k \geq 0 \}$ can be recognized by a 2D1CA. 
\end{lemma}


\begin{lemma}
	\label{lemma:division-test}
	For any given $ p>1 $, a 2D2CA $ \mathcal{D} $ with values $ M>0 $ and $0$ in its counters can test whether $ p $ divides $ M $ without moving the input head and, after testing, it can recover the values of the counters.
\end{lemma}

\section{Main results}

We start with the simulation of linearly bounded  multi--counter automata on unary languages and establish a direct connection with logarithmic-space unary languages. Secondly, we present the  simulation of  linear--space Turing machines on binary languages. Then we generalize this simulation for Turing machines  that use more space and for Turing machines without any resource bound. Thirdly, we present our quantum result. We finish the section comparing unary 2D1CAs and 2CAs.

\subsection{Simulation of multi-counter automata on unary alphabet}

We assume that all linearly bounded counters do not exceed the length of the input.
Let $ L \subseteq \{a\}^* $ be a unary language recognized by a 2D2CA $ \mathcal{M} $ with linearly bounded counters and  $ w=a^{n} $ be the given input that is placed on the input tape (between the two end-markers as $ \cent a^n \dollar $ and indexed from 0 to $ |w|+1 $). We can represent the configurations of $ \mathcal{M} $ on $ w $ with a state, an integer, and a Boolean variable as follows:
\begin{equation}
	\label{eq:configuration}
	(s,2^{i}3^{n-i}5^{c_1}7^{c_2},OnDollar),
\end{equation}
where 
\begin{itemize}
	\item $s$ is the current state,
	\item $ OnDollar = true $ means that the input head is on $ \dollar $,
	\item $ OnDollar = false $ means that the input head is on the $ i^{th} $ square, and,
	\item  $ c_1 $ (resp. $ c_2 $) represents the value of the first (resp. second) counter.
\end{itemize} 
By using $ OnDollar $ variable, we do not need to set $ i $ to $ (n+1) $ and this will simplify the languages that we will define soon. Note that we are using two exponents, i.e. $ 2^{i}3^{n-i} $, to store the position of the input head. In this way, we can implicitly store the length of the given input ($ n $).

\begin{lemma}
	\label{lemma:2d2ca-simulation}
	A 2D2CA, say $ \mathcal{M}' $, can simulate $ \mathcal{M} $ on $ w $ without using its input head, if its first counter is set to $ 2^{0}3^{n}5^{0}7^{0} $.
\end{lemma}
\begin{proof}
	For any $ p \in \{2,3,5,7\} $, $ \mathcal{M'} $, by help of the second counter, can easily increase the exponent of $ p $ by 1, test whether the exponent of $ p $  is zero or not, and decrease the exponent of $ p $ by 1 if it is not zero. Moreover, $ \mathcal{M'} $ can keep the value of $ OnDollar $, which is $ false $ at the beginning, by using its control states. Note that when the exponent of $ 3 $ is zero and the input head of $ \mathcal{M} $ is moved to the right, the value of $ OnDollar $ is set to $ true $; and, whenever the input head of $ \mathcal{M} $ leaves the right end-marker, the value of $ OnDollar $ is set to $ false $ again. During both operations, the exponents of $ 2 $ and $ 3 $ remain the same. Thus, $ \mathcal{M'} $ can simulate $ \mathcal{M} $ on $ w $ and it never needs to use its input head.
\qed\end{proof}
Note that, during the simulation given above, $ 2^i3^{n-i} $ is always less than $ 3^n $ for any $ i \in \{0,\ldots,n\} $, and so, the values of  both counters never exceed $ 3^n5^n7^n $.

Now, we build a 2D1CA, say $ \mathcal{M''} $, simulating the computation of $ \mathcal{M'} $ on some specific unary inputs. Let $ u \subseteq \{a\}^* $ be the given input. 
\begin{enumerate}
	\item $ \mathcal{M''} $ checks whether the input is of the form $ 3^n5^n7^n = 105^n $ for a non-negative integer $ n $ (Lemma \ref{lemma:UPOWER}). If not, it rejects the input.
	\item $ \mathcal{M''} $ sets its counter to $ 2^{0}3^{n}5^{0}7^{0} $ (Lemma \ref{lemma:divide-by-p}). Then, by using its input head as the second counter, it simulates $ \mathcal{M'} $ which actually simulates $ M $ on $ a^n $ (Lemma \ref{lemma:2d2ca-simulation}). $ \mathcal{M''} $ accepts (resp. rejects) the input if $ \mathcal{M} $ ends with the decision of ``acceptance'' (resp. ``rejection'').
\end{enumerate}
Thus, we can obtain that if $ L \subseteq \{a\}^* $ can be recognized by a 2D2CA with linearly bounded counters, then $ \{a^{105^n} \mid a^n \in L \} $ is recognized by a 2D1CA. Actually, we can replace $ 105 $ with $ 42 $ by changing the representation given in Equation \ref{eq:configuration} as:
\[
	(s,5^i7^{n-i}2^{c_1}3^{c_2},OnDollar),
\]	
where $ 5^i7^{n-i} $ is always less than $ 7^n $ for any $ i \in \{0,\ldots,n\} $.
\begin{theorem}
	If $ L \subseteq \{a\}^* $ can be recognized by a 2D2CA with linearly bounded counters, then $ \{a^{42^n} \mid a^n \in L \} $ is recognized by a 2D1CA.
\end{theorem}

Based on this theorem, we can easily show some languages recognized by 2D1CAs, e.g. 
\[
	\left\lbrace a^{42^{n^2}} \mid n \geq 0 \right\rbrace \mbox{ and } \left\lbrace a^{42^p} \mid p \mbox{ is a prime} \right\rbrace.
\]
We can generalize our result for 2D$k$CAs with linearly bounded counters in a straightforward way.
\begin{theorem}
	Let $ k>2 $ and $ p_1,\ldots,p_{k+1} $ be some different prime numbers such that one of them is greater than the $ (k+1)^{th} $ prime number. If $ L \subseteq \{a\}^* $ can be recognized by a 2D$k$CA with linearly bounded counters, then 
	\[ \left\lbrace a^{(p_1 \cdot p_2 \cdots p_{k+1})^n} \mid a^n \in L \right\rbrace \] is recognized by a 2D1CA.
\end{theorem}
\begin{proof}
	Let $ P = \{ p_1,\ldots,p_{k+1} \} $. Since one prime number in $ P $, say $ p_{k+1} $, is greater than the $ (k+1)^{th} $ prime number, there should be a prime number not in $ P $, say $ p_{k+1}' $, that is not greater than the $ (k+1)^{th} $ prime number. We can use the representation given in Equation \ref{eq:configuration} for a configuration of the 2D$k$CA:
\[
	\left(s, p_{k+1}^i(p_{k+1}')^{n-i} p_1^{c_1}p_2^{c_2}\cdots p_k^{c_k},OnDollar \right).
\]
$ p_{k+1}^i(p_{k+1}')^{n-i} $ is always less than $ p_{k+1}^n $, and so, the integer part of the configuration is always less than $ (p_1 \cdot p_2 \cdots p_{k+1})^n $. As described before, a 2D1CA can check whether the length of the input is a power of $ (p_1 \cdot p_2 \cdots p_{k+1}) $, and, if so, it can simulate the computation of the 2D$k$CA on the input. The 2D1CA needs to simulate $ k $ counters instead of 2 counters  but the technique is essentially the same.
\qed\end{proof}

The simulation given above can be easily generalized for nondeterministic, alternation, and probabilistic models. The input check and the initialization of the simulation are done deterministically. Therefore, the computation trees of the simulated and simulating machines have the same structure for the well-formed inputs, i.e. the inputs not rejected by the initial input check.

\begin{theorem}
	Let $ k \geq 2 $ and $ p_1,\ldots,p_{k+1} $ be some different prime numbers such that one of them is greater than the $ (k+1)^{th} $ prime number. If $ L \subseteq \{a\}^* $ can be recognized by a 2N$k$CA (resp. 2A$k$CA, bounded-error 2P$k$CA, or unbounded-error 2P$k$CA) with linearly bounded counters, then \[ \left\lbrace a^{(p_1 \cdot p_2 \cdots p_{k+1})^n} \mid a^n \in L \right\rbrace \] is recognized by a 2N1CA (resp. 2A1CA, bounded-error 2P1CA, or unbounded-error 2P1CA).
\end{theorem}

Now, we establish the connection with logarithmic-space unary languages.
The following two easy lemmas are a direct consequence
of the fact that, over unary alphabet, a linear bounded counter can be simulated
by the head position and vice versa.
\begin{lemma}
	\label{lemma:multihead}
	Any two-way automaton with $k$-heads on unary inputs can be simulated by a two-way automaton with $k$-linearly bounded counters, where $k>1$. 
\end{lemma}

The reverse simulation holds even on generic alphabets.

\begin{lemma}
	Any two-way automaton with $k$-linearly bounded counters can be simulated by a two-way automaton with $(k+1)$-heads, where $k>1$.
\end{lemma}
Both simulations work for deterministic, nondeterministic, alternating, and bounded- and unbounded-error probabilistic models. 
\begin{fact}
	\cite{Har72,Kin88,Mac97}
	The class of languages recognized by two-way multi-head deterministic, nondeterministic, alternating, bounded-error probabilistic, and unbounded-error probabilistic finite automata are
	\[
		\mathsf{L}, \mathsf{NL}, \mathsf{AL} (= \mathsf{P}), \mathsf{BPL}, \mbox{ and } \mathsf{PL},
	\]
	(deterministic, nondeterministic, alternating, bounded-error probabilistic, and unbounded-error probabilistic logarithmic space) respectively.
\end{fact}
Based on this fact, the last two lemmas, and the other results in this section, we can obtain the following theorem.
\begin{theorem}
	Let $L$ be any unary language in $\mathsf{L}$ (resp., $\mathsf{NL}$, $\mathsf{P}$, $\mathsf{BPL}$, and $\mathsf{PL}$). Then there is an integer $p$, product of some primes, such that
	\[
		\{ a^{p^n} \mid a^n \in L \}
	\]
	can be recognized by a 2D1CA (resp., 2N1CA, 2A1CA, bounded-error 2P1CA, and unbounded-error 2P1CA).
\end{theorem}

\subsection{Simulation of Turing machines on binary and general alphabets}
\label{sec:simulation-TM}

Let $ \mathcal{N} $ be a single-tape single-head DTM (deterministic Turing machine) working on a binary alphabet $ \Sigma=\{a,b\} $. Note that its tape alphabet also contains the blank symbol $ \# $. We assume that the input is written between two blank symbols for DTMs. We define some restrictions on $ \mathcal{N} $:
\begin{itemize}
	\item There can be at most one block of non-blank symbols.
	\item The tape head is placed on the right end-marker at the beginning of the computation which makes easier to explain our encoding  used by the 2D3CA given below. Note that this does not change the computational power of the DTMs.
\end{itemize}
A configuration of $ \mathcal{N} $ on a given input, say $ w \in \{a,b\}^* $, can be represented as
$
	u s v,
$
where $ uv \in \# \{a,b\}^{*} \# $ represents the current tape content and $ s $ is the current state. Moreover, the tape head is on the last symbol of $ \#u $. The initial configuration is $ \#w\#s_1 $, where $ s_1 $ is the initial state. Here $ v $ is the empty string. Note that $ u $ can never be the empty string.

By replacing $ a $ with 0, and $ b $ and $ \# $ with 1s, we obtain a binary number representation of $ u $ and $ v $ -- we will denote these binary numbers by $ u $ and $ v $, respectively. Now, we give a simulation of $ \mathcal{N} $ by a 2D3CA, say $ \mathcal{N'} $, on $ w $.\footnote{We refer the reader to \cite{vEB90} for a general theory of simulations.} $ \mathcal{N'} $ does not have a tape but can simulate it by using three counters. During the simulation, $ \mathcal{N'} $ keeps $ u $ and $ v^r $ on its two counters. If $ \mathcal{N'} $ knows the state and the symbol under head, it can update the simulating tape. $ \mathcal{N'} $ can keep the state of $ \mathcal{N} $ by its internal states and can easily check whether:
\begin{itemize}
	\item $ u $ equals to 1 or is bigger than 1;
	\item $ v^r $ equals to 0, equals to 1, or bigger than 1;
	\item the last digit of $ u $ is zero or one; and
	\item the last digit of $ v^r $ is zero or one.
\end{itemize}



Based on these checks, $ \mathcal{N'} $ can simulate the corresponding change on the tape (in a single step of $ \mathcal{N} $)  with the values of the counters.

By a suitable encoding, two counters can simulate $ k>2 $ counters. Let $ p_1, p_2, \ldots, p_k $ be co-prime numbers. The values of all $ k $ counters, say $c_1,c_2,\ldots,c_k$, can be represented as $ p_1^{c_1}p_2^{c_2}\cdots p_k^{c_k} $. A counter can hold this value, and, by using the second counter, 
$ \mathcal{N'} $ can check if $ c_i $ is equal to zero  (Lemma \ref{lemma:division-test}) and it can simulate an appropriate increment/decrement operation on $ c_i $, where $ 1 \leq i \leq k $. Therefore, we can  conclude that a 2D2CA, say $ \mathcal{N''} $, can simulate $ \mathcal{N} $ on $ w $ by using prime numbers $ \{2,3,5\} $ for encoding, if its first counter is set to $ 3^{1w1} $. Here $ \mathcal{N''} $ can use the exponents of $ 3 $ and $ 5 $ for keeping the content of the tape and the exponent of $ 2 $ to simulate the third counter. 

Let's assume that $ \mathcal{N} $ uses exactly $ |w|+2 $ space, i.e. the tape head never leaves the tape squares initially containing $\#w\#$. That is, the binary value of the tape is always less than twice of $ 1w1 $, which is $1w10$, during the whole computation. Then the values of the counters can never exceed $ 5^{1w10}2^{1w10} $ or $ 25^{1w1}4^{1w1} $, where the exponents are the numbers in binary. Note that the whole tape is kept by the exponent of $ 3 $ and $ 5 $, and so, their product is always less than $ 5^{1w10} $.

\begin{theorem}
\label{theorem:dtmsimulation}
	If $ L \subseteq \{a,b\}^* $ can be recognized by a DTM, say $\mathcal{N}$, in space $ |w|+2 $ with binary work alphabet, then \[  \left\lbrace  a^{100^{1w1}} \mid w \in L \right\rbrace \] can be recognized by a 2D1CA, say $\mathcal{N}'''$. 
\end{theorem}
\begin{proof}
	$\mathcal{N}'''$ rejects the input if it is not of the form $ \{a^{100^n}\} $ (Lemma \ref{lemma:UPOWER}), where $ n>0 $. Then, it sets its first counter to $ 4^n $ (Lemma \ref{lemma:divide-by-p}). $\mathcal{N}'''$ rejects the input, if $ n $ is not of the form $ 1w1 $ for some $ w \in \{a,b\} $. We know that a 2D2CA can easily do this check if one of its counter is set to $n$, i.e. it needs to check $ n $ is odd and $ n \notin \{0,1,2\} $. So, $ \mathcal{N'''} $ can implement this test by using its input head as the second counter.
	
As described above, if its first counter is set to $ 3^{1w1} $, the 2D2CA $ \mathcal{N''} $ can simulate $\mathcal{N}$ on a given input $w$. Due to the space restriction on $ \mathcal{N} $, we also know that the counter values (of the 2D2CA's) never exceed $ 100^{1w1} $. So, $\mathcal{N'''}$ needs only to set its counter value to $ 3^{1w1} $.  $\mathcal{N'''}$ firstly sets its counter to $ 4^{1w1}3^0 $, and then transfers $ 1w1 $ to the exponent of $3$.
\qed\end{proof}

Remark that the language recognized by $\mathcal{N'''}$ can also be represented as \[ \left\lbrace a^{10^{1w10}} \mid w \in L \right\rbrace. \] This representation is more convenient when considering DTMs working on bigger alphabets.

\begin{corollary}
	Let $ k>2 $ and $ L \subseteq \{a_1,\ldots,a_k\}^* $ be a language recognized by a DTM in space $ |w|+2 $ with a work alphabet having $ k' \geq k $ elements. Then \[  \left\lbrace a^{10^{1w10}} \mid w \in L \right\rbrace \] can be recognized by a 2D1CA, where $ w \in \{a_1,\ldots,a_k\}^* $ and $1w10$ is a number in base-$k'$.
\end{corollary}
\begin{proof}
	The proof is almost the same by changing base-$2$ with base-$k'$. Additionally, the 2D1CA needs to check whether each digit of $w$ is less than $ k $.
\qed\end{proof}

We know that 2D1CAs can recognize $ \mathtt{POWER} = \left\lbrace  a^nba^{2^n} \mid n > 0 \right\rbrace $ \cite{Pet94}. Therefore, by using a binary encoding, we can give a simulation of exponential space DTMs where the exponent is linear. Here, the input is supposed to be encoded into the exponent of the first block of $a$'s and the working memory in the second block of $a$'s.

\begin{theorem}
	Let $ k>2 $ and $ L \subseteq \{a_1,\ldots,a_k\}^* $ be a language recognized by a DTM in space $ 2^{|w|} $ with a work alphabet having $ k' \geq k $ elements. Then 
	\[ \left\lbrace a^{10^{x}}ba^{2^{\left( 10^{x} \right)}} \mid x=1w10 \mbox{ and } w \in L \right\rbrace \] can be recognized by a 2D1CA, where $ w \in \{a_1,\ldots,a_k\}^* $ and $x=1w10$ is a number in base-$k'$.
\end{theorem}

We can generalize this result for any arbitrary space-bounded DTMs. It is not hard to show that, for any $z>1$, 2D1CAs can recognize 
\[
	\mathtt{POWER(z)} = \left\lbrace a^nba^{exp(n)}ba^{exp^2(n)}b \cdots b a^{exp^z(n)}  \mid n > 0 \right\rbrace.
\]

\begin{corollary}
	Let $z>1$ and $ k>2 $ and $ L \subseteq \{a_1,\ldots,a_k\}^* $ be a language recognized by a DTM in space $ exp^z(|w|) $ with a work alphabet with $ k' \geq k $ elements. Then 
	\[ 
	 	\left\lbrace	a^{10^{x}}b a^{exp \left( 10^{x} \right)} b a^{exp^2 \left( 10^{x} \right)} b \cdots b a^{exp^z \left( 10^{x} \right)} \mid x=1w10 \mbox{ and } w \in L \right\rbrace 
	 \] can be recognized by a 2D1CA, where $ w \in \{a_1,\ldots,a_k\}^* $ and $x=1w10$ is a number in base-$k'$.
\end{corollary}

Note that, similar to the previous section, all of the above results are valid if we replace deterministic machines with nondeterministic, alternating, or probabilistic ones.

Now, we present a more general result.

\begin{theorem}
	\label{theorem:recenum}

Let $ L $ be a recursive enumerable language and $ \mathcal{T} $ be a DTM recognizing it (note that $ \mathcal{T} $ may not halt on some non-members). The language
\[
	L_{\mathcal{T}} = \left\lbrace a^{2^{1w}3^{S(w)}} \mid w \in L \right\rbrace,
\]
where $ S(w) $ is a sufficiently big number that depends on $w$,
can be recognized by a two way deterministic one counter automaton $\mathcal{D}$.
\end{theorem}
\begin{proof}
We use a slight variation of the 2DCA simulation of a DTM given above. Informally the $ a^{3^{S(w)} } $ part of the input gives $ \mathcal{D} $ enough space to complete its simulation, i.e. decide the membership of $w \in L$ using its head position as a second counter, being sure that its value never exceeds
the size of the input.

First we show that if $S(w)$ is large enough then a 2D1CA $\mathcal{D}$ can recognize the language $L_{\geq \mathcal{T}}$:
\[
	L_{\geq \mathcal{T}} = \left\lbrace a^{2^{1w}3^{k}} \mid w \in L \mbox{ and } k \geq S(w) \right\rbrace .
\]
$\mathcal{D}$ checks that the input is 
in the correct format $a^{2^{1w} 3^{k}}$  (Lemma~\ref{lemma:coprimespow}), then it simulates $\mathcal{T}$ 
on $w$ like showed in the proof of Theorem~\ref{theorem:dtmsimulation}.
During its computation, if $\mathcal{D}$ reaches the right end-marker, 
then it stops and rejects.

Suppose that on input $w$ the Turing machine $\mathcal{T}$ does not halt:
 it visits an infinite number of empty cells of its tape
or it enters an infinite loop. In both cases, the value of $S(w)$ is irrelevant, and $\mathcal{D}$ will never accept the input $a^{2^{1w} 3^{S(w)}}$:
in the first case for all values of $S(w)$ $\mathcal{D}$ will hit
the right end-marker and will reject; in the second case, if $S(w)$ is too
low and $\mathcal{D}$ has not enough space to simulate $\mathcal{T}$ in the
loop area of the tape it will hit the right end-marker and reject, if $S(w)$
is large enough, $\mathcal{D}$ will also enter the endless loop.

Now suppose that the Turing machine $\mathcal{T}$ accepts (resp. rejects) $w$ then
there are two possibilities: {\emph a)} during its computation the 2D1CA (that
uses the head position like a second counter) never reaches the right end-marker; in this
case it can correctly accept (resp. reject) the input; or {\emph b)} during its computation the 2D1CA
reaches the right end-marker (informally it has not enough space) and cannot correctly decide
the membership of $w \in L$; but in this case we
are sure that there exists a larger value $S(w) = s' > s$ that assures enough
space to end the computation. Also for every $k \geq S(w)$, $\mathcal{T}$ will correctly
accept each string in \[ \left\lbrace a^{2^{1w}3^{k}} \mid w \in L \right\rbrace. \]

We can slightly modify $\mathcal{D}$ and narrow the language it recognizes
to exactly $ L_{ \mathcal{T} } $, i.e. 
making it accept each string in:
\[ 
\left\lbrace a^{2^{1w}3^{S(w)}} \mid w \in L \right\rbrace , 
\]
but reject each string in:
\[ 
\{ a^{2^{1w}3^{i}} \mid w \in L \mbox{ and } i \neq S(w) \}. 
\]
We can divide the natural numbers as follows:
$$[0,3N)\quad [3N,9N)\quad [9N,27N) \quad ... \quad [3^{k-1} N, 3^k N)\quad [3^k N, 3^{k+1} N)\quad ...$$
Let $M$ be the maximum value of the second counter of $\mathcal{D}$ during
the simulation of the DTM on $w$ (for each member of $L$, such value exists).
$M$ must be in one of the above intervals, let's say in $[3^{k-1} N, 3^{k} N)$. It is obvious that $3M$ must be in the next interval $[3^{k} N, 3^{k+1} N)$.

The second counter can use the set $\{+3,0,-3\}$ instead of
$\{+1,0,-1\}$ for update operations, i.e. the head moves three
steps left or right instead of a single step, and 
using the internal states we can allow it to
exceed the input length up to three times its value:
when the head reaches the right end-marker it can keep track 
that it has made one \emph{``fold''} and continues moving towards the left;
thereafter, if it reaches the left end-marker, it records that it has made two folds and continues move rightward, and so on. When, after a fold, it hits the last end-marker again it can decrease the number of folds and continue.
Let $\mathsf{FOLD} \in \{0,1,2,3\}$ store the number of folds. When $\mathsf{FOLD}$ becomes 3, then the 2D1CA rejects the input immediately, i.e. the counter value reaches the value of three times of the input length.

On input $a^{ 2^{1w} 3^k}$, the second counter, that changes its value by $\{+3,0,-3\}$, will  exceed $3^k N$ but will
never try to exceed $3^{k+1} N$ ($3^k N \leq 3M < 3^{k+1} N$). So, $\mathsf{FOLD}$ must be 1 at least once and
never becomes 3. Therefore, the 2D1CA accepts the input if the
simulation ends with the decision of ``acceptance'' and $\mathsf{FOLD}$ takes a
non-zero value at least once but never takes the value of 3.

If the input is $a^{ 2^{1w} 3^{k-i}}$, for some positive integer $i$, then
the second counter must need to exceed $3^k N$, so, the $\mathsf{FOLD}$ value takes
3 before simulation terminates and the 2D1CA rejects the input.

If the input is $a^{ 2^{1w} 3^{k+i}}$, for some positive integer $i$, then
the second counter can be at most $3^{k+1}-1$, so the $\mathsf{FOLD}$ value
never takes the value of 1 during the simulation and the 2D1CA rejects
the input.
Thus, we can be sure that such $k$ has a minimum value and it corresponds
to $S(w)$ in the language $ L_{\mathcal{T}} $.
\qed\end{proof}

Note that if the language $ L $ recognized by $\mathcal{T}$ is recursive, then the same 2D1CA $\mathcal{D}$ described in Theorem~\ref{theorem:recenum} is a decider for $L_{\mathcal{T}}$. 

\subsection{A quantum simplification}
Ambainis and Watrous \cite{AW02} showed that augmenting a two-way deterministic finite automata (2DFAs) with a fixed-size quantum register\footnote{It is a constant-size quantum memory whose dimension does not depend on the input length. The machine can apply to the register some quantum operators (unitary operators, measurements, or superoperators)  determined by the classical configuration of the machine. If the operator is a measurement or a superoperator, then there can be more than one outcome and the next classical transition is also determined by this outcome, which makes the computation probabilistic. However, as opposed to using a random number generator, the machine can store some information on the quantum register and some pre-defined branches can disappear during the computation due to the interference of the quantum states, which can give some extra computational  power to the machine.} makes them more powerful than 2DFAs augmented with a random number generator. Based on a new programming technique given for fixed-size quantum registers \cite{YS10B}, it was shown that 2D1CAs having a fixed-size quantum register can recognize 
$ \{ a^{n3^n} \mid n \geq 1 \},~ \{ a^{2^n3^{2^n}} \mid n \geq 1 \}, $ or any similar language by replacing bases 2 or 3 with some other integers for any error bound \cite{Yak13D,KY14A}. Therefore, we can replace binary encoding with a unary one for Theorem 6 by enhancing a 2D1CA with a fixed-size quantum register.

\begin{theorem}
	Let $ k>2 $ and $ L \subseteq \{a_1,\ldots,a_k\}^* $ be a language recognized by a DTM in space $ 3^{|w|} $ with a work alphabet having $ k' \geq k $ elements. Then 
\[
	\left\lbrace a^{10^{x}3^{\left( 10^{x} \right)}} \mid x=1w10 \mbox{ and } w \in L \right\rbrace
\] can be recognized by a 2D1CA augmented with a fixed-size quantum register for any error bound, where $x=1w10$ is a number in base-$k'$ and $ w \in \{a_1,\ldots,a_k\}^* $.
\end{theorem}
\begin{proof}
	Here the input check can be done by the help of the quantum register by using the corresponding quantum algorithms given in \cite{Yak13D}. Then, our standard deterministic simulation is implemented.
\qed\end{proof}

\subsection{Unary 2D1CAs versus two-counter machines}

Minsky \cite{Min61} showed that, for any given recursive language $L$ defined over $\mathbb{N}$, \[ \mathtt{UMINSKY}(L) = \{ a^{2^x} \mid x \in L \} \] can be recognized by a 2CA\footnote{The definition used by Minsky is a little different than ours but they are equivalent.}.
It is clear that $ \mathtt{UMINSKY}(L) $ is recursive enumerable if and only if $ L $ is recursive enumerable. Moreover, any language $L$ not recognizable by any $s(n)$-space DTM, $ \mathtt{UMINSKY}(L) $ cannot be recognized by any $ \log(s(n)) $-space DTM, for any $ s(n) \in \Omega(n) $. On the other hand, any language recognized by a 2D1CA is in $\mathsf{L} $ (see Footnote \ref{foot:2D1CA}). Therefore, there are many recursive and non-recursive languages recognized by 2CAs but not by 2D1CAs. 

Neverthless we believe that 2CAs and unary 2D1CAs are incomparable, i.e.
there also exist languages recognizable by a 2D1CA but not by any 2CA.  Let $k>1$, $\Sigma = \{ a_0,\ldots,a_{k-1} \}$ be the alphabet, and $ r_k : \mathbb{N} \rightarrow \Sigma^* $ be a function mapping $ n = (d_{l}\cdots d_1 d_0)_k $, $k$-ary representation of $n$, to
\[
	r_k(n) = \left\lbrace \begin{array}{lr} a_{d_0} a_{d_1}\cdots  a_{d_{l}}, & \mbox{if } n>0 \\ 
	\varepsilon, & \mbox{ if } n = 0
	\end{array} \right. .
\]
Both 2D1CAs and 2CAs can calculate $r_k(n)$ symbol by symbol on the input $a^n$, and the following is immediate:

\begin{lemma}
	\label{lem:regular-exponent}
If $R$ is a regular languages over an alphabet of $k$ symbols, then a 2D1CA can decide the language $L = \{a^n \mid r_k(n) \in R \}$. 
\end{lemma}
\begin{proof}
	Suppose that $\mathcal{F}$ is a DFA that decides $R$; after transferring the input to the counter, a 2D1CA can calculate incrementally the digits $d_0, d_{1},...,d_{l}$ up to the final fixed digit: it repeatedly divides the counter by $k$, and $d_i$ is the remainder of the division; so it can simulate the transition of $\mathcal{F}$ on symbol $a_{d_i}$, and accept or reject accordingly when it reaches the last digit.
\qed\end{proof}

Hence both 2D1CAs and 2CAs can recognize the whole class of unary languages:
$$\begin{array}{c}
\mathcal{C} = \{ L \mid L = \{a^n \mid r_k(n) \in R \} \text{ and } R \text{ is a regular language}\\
\text{over an alphabet of size } k \}
\end{array}$$
As an example the family of
unary non regular languages $\{ a^{2^n} \}$ is contained in $\mathcal{C}$.
But, we conjecture that the following language cannot be recognized by 2CAs:
\[
	L_{\oplus} = \{ a^n \mid |r_2(n)|+|r_3(n)| \equiv 0 \mod 2 \},
\]
i.e. the binary representation and the ternary representation of $n$ have both even or odd length.
A 2D1CA can easily decide $L_{\oplus}$: after calculating if the length of the binary representation of $n$ is odd or even, it can recover the input using the tape
endmarkers, and then check if the length of the ternary representation of $n$ is the same.
But there is no way for a 2CA to recover the input, so it should
calculate the binary and ternary representations of $n$ in parallel, which
seems impossible.

 ~\\ \noindent\textbf{Acknowledgements.} 
We thank Alexander Okhotin, Holger Petersen, and Klaus Reinhardt for their answers to our questions on the subject matter of this paper. We also thank anonymous referees for their very helpful comments. 

\bibliographystyle{plain}
\bibliography{tcs}

\end{document}